\documentclass[twocolumn,prx,aps,showpacs,superscriptaddress,floatfix,nofootinbib,10pt]{revtex4}

\usepackage{amsfonts,amsmath,amssymb}
\usepackage{graphicx}
\usepackage{color}
\usepackage[normalem]{ulem}
\usepackage{amsfonts}
\usepackage[toc,page]{appendix}
\usepackage[ocgcolorlinks,colorlinks=true,linkcolor=blue,citecolor=red]{hyperref}
\usepackage{latexsym}
\usepackage{amsfonts}
\usepackage{algpseudocode}
\usepackage{amsthm}
\usepackage{mathrsfs}
\usepackage{natbib}
\usepackage{color,verbatim}
\usepackage{psfrag}
\usepackage{blindtext}
\usepackage{enumitem}

\newtheorem{theorem}{Theorem}

\DeclareMathOperator{\Tr}{Tr}

\newcommand{\be}{\begin{eqnarray}}
\newcommand{\ee}{\end{eqnarray}}

\begin{document}
\title{Higher dimensional communication complexity problems: classical protocols vs quantum ones based on Bell's Theorem or prepare-transmit-measure schemes}
\author{Armin Tavakoli}
\affiliation{Department of Physics, Stockholm University, S-10691 Stockholm, Sweden.}
\affiliation{Groupe de Physique Appliqu\'ee, Universit\'e de Gen\'eve, CH-1211 Gen\'eve, Switzerland}
\author{Marek \.Zukowski}
\affiliation{Institute of Theoretical Physics and Astrophysics, Uniwersytet Gda\'nski, PL-80-952 Gda\'nsk, Poland.}


\date{\today}


\begin{abstract}
Communication complexity problems (CCPs) are tasks in which separated parties attempt to compute a function whose inputs are distributed among the parties. Their communication is limited so that not all inputs can be sent. We show that broad classes of Bell inequalities can be mapped to CCPs and that a quantum violation of a Bell inequality is a necessary and sufficient condition for an enhancement of the related CCP beyond its classical limitation. However, one can implement CCPs by transmitting a quantum system, encoding no more  information than is allowed in the CCP, and extract information by performing measurements. We show that for a large class of Bell inequalities, the improvement of the CCP associated to a quantum violation of a Bell inequality can be no greater than the improvement obtained from quantum prepare-transmit-measure strategies.    

\end{abstract}



\maketitle

\textit{Introduction.---} Bell's theorem asserts that measurements on separated entangled quantum states can give rise to outcome correlations that have no local realistic model \cite{Bell64}. This fact can be used to break classical limits in communication complexity problems (CCPs) \cite{BD99}. However, quantum protocols for CCPs violating classical bounds, that are based on prepare-transmit-measure schemes involving just a single quantum system are also possible \cite{GALVAO}. This can be certified by a violation of an inequality bounding the strength of the classical counterpart of such a protocol.

Nonclassical features of various quantum predictions are an essential tool in many quantum information tasks such as (semi) device-independent cryptography \cite{AG06, PB11}, randomness generation \cite{PAM10, LP12, MT16} and dimension witnesses \cite{BP08, GB10}. However, in terms of studying the fundamental physical phenomena, correlations due to the entanglement of two or more systems have  been given significantly more attention than those obtained from preparing and measuring a single quantum system. Indeed, little is known about the relation between the strength of the two general types of nonclassical correlations enabled by quantum theory, and their comparative applicability in quantum protocols violating classical bounds in information processing tasks.  

Here, we aim to construct a game-theoretic framework in which one can, on equal footing, compare  the communication complexity reduction power of entanglement based protocols  and single quantum system approaches. For this purpose, we will use a class of information-theoretic games related in fact to CCPs.

 In CCPs, a number of parties, say $N$, attempt to jointly compute a task function $f(X_1,\ldots,X_N)$. However, the input $X_i$ is only known to party  $i$. The task is to maximize the probability of one party to correctly compute $f$ when the amount of allowed communication between the parties is limited by some rule, which does not allow to transmit all data contained in any $X_i$. On one hand, since single system protocols are based on measurements on a transmitted quantum system of a specific dimension $d$, which constrains its information carrying capacity to $\log{d}$ bits, appropriate  CCPs  are a natural habitat in which the quantum strength of such CCP protocols can be studied. On the other hand, Bell inequalities are known to exhibit links to game theory \cite{BL13}.  

The relation between CCPs and  correlations due to entanglement has been extensively studied \cite{BD99, BZ, BZ04}, but initial steps in the direction of using these games as a framework to study both types of quantum CCP protocols has only recently been taken in Refs.\cite{PZ10, magic7, TMP16}.

We will show that for every bipartite Bell inequality, we can formulate a  CCP such that the reduction of communication complexity obtained from using classical communication assisted by correlations due to shared entanglement directly corresponds to the ability of quantum theory to violate the original Bell inequality. However, the CCP can also be implemented in quantum theory by the preparation, transmission and measurement of a single quantum system. 
Using such CCPs as a framework for both types of quantum resources, we will show that for large classes of Bell inequalities, correlations due to measurements on entangled states cannot beat the performance of quantum prepare-transmit-measure protocols.

\textit{The studied class of Bell inequalities.---}
In a bipartite Bell inequality, observers Alice and Bob perform measurements $x\in \{0,\ldots,m_A-1\}$ and $y\in\{0,\ldots,m_B-1\}$ repsectively, with a distribution $p(x,y)$. Each measurement has an outcome $a,b\in\{0,\ldots, d-1\}$ respectively. Such Bell inequalities can in a general way be written as 
\begin{equation}\label{Bell1}
\sum_{x,y}p(x,y)\sum_{a,b=0}^{d-1}\sum_{k=0}^{K}c^k_{ab|xy}P(a,b|x,y)\leq B.
\end{equation}
$B$ is the classical bound, $c^k_{ab|xy}$ are real numbers, and $k$ is an  index with some range $k\in\{0,\ldots,K\}$ for some natural number $K$. This index will allow us to put the inequalities in a form which generalizes the form of the CGLMP inequalities \cite{CGLMP02}. Note that we can without loss of generality assume that $\forall a,b,x,y$, there exists at most one $k\in\{0,\ldots, K\}$ such that $c^k_{ab|xy}\neq 0$. To see this, simply note that if $c^k_{ab|xy}$ and $c^{k'}_{ab|xy}$ with $k\neq k'$ were both non-zero, then in Eq.\eqref{Bell1} we would encounter the terms $c^k_{ab|xy}P(a,b|x,y)+c^{k'}_{ab|xy}P(a,b|x,y)=(c^k_{ab|xy}+c^{k'}_{ab|xy})P(a,b|x,y)$ where the pre-factor is again just some real number.

The Bell inequalities of our interest have the following structure. Firstly, we draw inspiration from a variety of known Bell inequalities \cite{TMP16, LL10, CGLMP02, SA16, BK06, TZ16} in which correlations between Alice's and Bob's local outcomes are quantified using their sum $a+b\mod{d}$. With this in mind, for any given pair of meausrements $(x,y)$, we construct the set $S_{xy}=\{\forall (a,b,k) \text{ such that }c^k_{ab|xy} \text{is defined}\}$. We we require that $S_{xy}$ admits a partition of the form $\{S_{xy}^{i,k}\}_{i,k}$ for $i=1,\ldots, N$ and $k=0,\ldots, K$, for some integers $N,K$, with $S_{xy}^{i,k}=\{(a,b)| a+b=F_{xy}^i(k) \mod{d}\}$ for some functions $F_{xy}^{i}(k)$ i.e.,
\begin{gather}\label{Bell2}\nonumber
	\forall x\forall y: S_{xy}=\bigcup_{i=1}^{N}\bigcup_{k=0}^{K} S_{xy}^{i,k}, \\
	\text{ and  }
S_{xy}^{i,k}\cap S_{xy}^{i',k'}=\emptyset \text{ for } (i,k)\neq (i',k') .
\end{gather}
Remark: since the  sets $S_{xy}^{i,k}$ are disjoint for $i=1,\ldots, N$ and $k=0,\ldots, K$ it follows that there can be no set $(a,b)$ that simultaneously satisfies both $a+b=F_{xy}^j(k)$ and $a+b=F_{xy}^{j'}(k')$ for $(j,k)\neq (j',k')$. This implies that the range of $F_{xy}^j$ is disjoint with that of $F_{xy}^{j'}$ for $j\neq j'$. Also, since $a+b\mod{d}$ can have at most $d$ different values it follows that $(K+1)N\leq d$.

Secondly, we restrict the structure of $c^k_{ab|xy}$ such that we later can make the connection to a related family  of CCPs. To see why this restriction is necessary, we remind ourselves that in CCPs Alice and Bob attempt to compute the value of some function from which they earn some payoff. Although the local outcomes produced from measurements on a (perhaps) entangled state may assist Alice and Bob in performing the computation, the values of these local outcomes are per se of no interest in the CCP. 
Therefore, we require that the Bell inequality is such that the same coefficient $c^k_{ab|xy}$ is assigned to any pair $(a,b)\in S_{xy}^{i,k}$, i.e., we may write $c^k_{ab|xy}=c^{i,k}_{xy}$.  

Thus, the Bell inequalities we will consider are written  
\begin{equation}\label{Bell3}
I^{bell}\equiv \sum_{x,y}p(x,y)\sum_{k=0}^{K}\sum_{i=1}^{N}c^{i,k}_{xy}P_{xy}(a+b=F_{xy}^i(k))\leq B.
\end{equation}
By $\mathscr{B}$, we will denote the some arbitrary  Bell inequality of this form. This class  can be viewed as a generalization of the inequalities considered in Ref.\cite{BZ},  from two-outcome to many-outcome Bell scenarios.

\textit{Representing Bell inequalities $\mathscr{B}$ as payoff bounds for classical CCPs.---} 
Consider the following family of CCPs, which we label $\mathbb{G}_{\mathscr{B}}$. Alice is given one input $x_0\in\{0,\ldots, d-1\}$ with $p(x_0)=1/d$, and one input $x\in\{0,\ldots, m_A-1\}$, whilst Bob receives one input $y\in\{0,\ldots, m_B-1\}$. The inputs $x$ of Alice and $y$ of Bob are distributed according to a  joint probability distribution $p(x,y)$. There is a communication channel from Alice to Bob over which Alice may send at most $\log d$ bits of information in form of a message $m(x_0,x)\in\{0,\ldots,d-1\}$. Having received the message, Bob outputs his guess  $G(y,m)\in\{0,\ldots, d-1\}$. If $G$ coincides with the value of one of the functions $f_{i,k}(x_0,x,y)=x_0+F_{xy}^i(k) \mod{d}$, then Alice and Bob jointly earn a payoff $c^{i,k}_{xy}$. The average earned payoff in $\mathbb{G}_{\mathscr{B}}$ is
\begin{multline}\label{perf}
	I_{\mathbb{G}_{\mathscr{B}}}^{cc}=\\\frac{1}{d}\sum_{x_0,x,y}p(x,y)\sum_{k=0}^{K}\sum_{i=1}^{N}c^{i,k}_{xy}P_{xyk}(G=f_{i,k}(x_0,x,y)).
\end{multline}

In a quantum version of such a CCP, to assist Alice's and Bob's attempts to perform optimally, they may perform measurements on their subsystems in an entangled state. Alice performs a local measurement of a setting labeled by $x$ and obtains the outcome $a\in\{0,\ldots, d-1\}$. Similarly, Bob performs a local measurement labeled by $y$ and obtains the outcome $b\in\{0,\ldots,d-1\}$.  Alice sends a message which depends on $x_0$ and $a$.  The other method is that Alice sends to Bob a quantum system of dimension $d$ in a state which depends on $x_0$ and $x$, upon which Bob performs a measurement of his choice, and somehow produces a guess.

There are many possible ways of implementing  $\mathbb{G}_{\mathscr{B}}$ by choosing different ways of coding the message $m$ and outputting the guess $G$. However, we shall limit the strategies under consideration to only such in which Bob's guess is of the form $m + b(y) \mod{d}$.
 In particular, we call any strategy linear, both in the case of classical and entanglement assisted CCPs, if $m=x_0+a(x)\mod{d}$.  Any other strategy of Alice we call nonlinear. We shall now state and prove a theorem about the optimality of such linear strategies in  $\mathbb{G}_{\mathscr{B}}$.


\begin{theorem}\label{t1}
	The optimal performance in classical $\mathbb{G}_{\mathscr{B}}$ is achieved with a linear strategy. Moreover, the performance of any nonlinear strategy is a probabilistic mixing of the performances of linear strategies. 
\end{theorem}

\begin{proof}

We first re-write our Bell inequalities in Eq.\eqref{Bell3}.  The discrete Fourier transform of $P(a+b|x,y)$ can be defined as $E(l|x,y)=\sum_{z=0}^{d-1}\omega^{lz}P(a+b=z|x,y)$ where $\omega=e^{\frac{2\pi i}{d}}$. Its inverse reads 
\begin{equation}
P(a+b=z|x,y)=\frac{1}{d}\sum_{l=0}^{d-1}\omega^{-lz}E(l|x,y).
\end{equation}
Therefore we get
\begin{equation}
P(a+b=F_{xy}^i(k)|x,y)=\frac{1}{d}\sum_{l=0}^{d-1}\omega^{-lF_{xy}^i(k)}E(l|x,y).
\end{equation}
By direct insertion into Eq.\eqref{Bell3}, we may write any Bell inequality $\mathscr{B}$ on the form
\begin{equation}
I^{bell}\equiv \sum_{x,y} \frac{p(x,y)}{d}\sum_{l=0}^{d-1}\sum_{i=0}^{N}\sum_{k=0}^{K}c^{i,k}_{xy}\omega^{-lF_{xy}^i(k)}E(l|x,y)\leq B.
\end{equation}
Next, notice that $E(l|x,y)$ is the average value of the products of the local results, each represented by specific powers of $\omega$, for local settings $x$, $y$.  Namely  $E(l|x,y)
=\langle \omega^{la}\omega^{lb} \rangle_{x,y} $. Thus, for each $l$ we have  a different form of correlation function.

Having written $\mathscr{B}$ in terms of correlators, it is now straightforward to write down the performance \eqref{perf} in $\mathbb{G}_{\mathscr{B}}$ in this terminology. The property of $d$-th roots of unity $\sum_{l=0}^{d-1}\omega^{l(z-q)}=d\delta_{z,q}$, where $z, q$ are integers, allows one to put the logical value of question of whether a guess $G$ equals to $f_{i,k}$ in the form of $\frac{1}{d}\sum_{l=0}^{d-1}\omega^{l(G-f_{i,k})}.$ Thus the  payoff of a $\mathbb{G}_{\mathscr{B}}$ class game, if the answer is $G$, is given by: 
\begin{multline}\label{CCP-PAYOFF}
I^{cc}_{\mathscr{B}}=\frac{1}{d}\sum_{x,y}p(x,y)\\
\sum_{x_0=0}^{d-1}\sum_{l=0}^{d-1}\sum_{i=1}^{N}\sum_{k=0}^{K}c^{i,k}_{xy}\omega^{-lF_{xy}^i(k)}\omega^{-lx_0}\tilde{G}^*(x_0,x,y)^l,
\end{multline}
where  $\tilde{G}=\omega^{G}$ is represents guess output by Bob, transformed into a power of $\omega$.   
Representation of the guess and the message in form of powers of $\omega$ will play an important technical role in what follows.

Assume now that Alice and Bob apply some general strategy in $\mathbb{G}_{\mathscr{B}}$, i.e., $G=G(m(x_0,x),y)=m(x_0,x)+b(y)$.
The guess $\tilde{G}_{xy}(x_0)=\omega^{G_{xy}(x_0)}=\tilde{M}_{x}(x_0)\omega^{b(y)}$  where $\tilde{M}_{x}(x_0)=\omega^{m(x_0,x)}$ is aways equal to some integer power of $\omega$.
 We shall analyze function $\tilde{M}$ by treating $x,$ as an index, for fixed values of which $\tilde{M}_{x}(x_0)$ is a function of $x_0$ only.

Notice that the part of the expression  in (\ref{CCP-PAYOFF}) which depends on $x_0$ is $\frac{1}{d}
\sum_{x_0=0}^{d-1}\omega^{-lx_0}\tilde{M}_{x}^*(x_0)^l.$   We see that  we have here the $l$-th value of a discrete Fourier transform of $(\tilde{M}_{x})^l$.
As we shall see below Fourier transforms of powers of  functions which can have values only in the form of powers of $\omega$, have very specific properties.

Take a function $B(x_0)$, where $x_0=0,1,...,d-1$, such that its values are only in the set the set $\omega^0, \omega, ...,\omega^{d-1}$. 

\textbf{Lemma 1.} The discrete Fourier transform  of $B$, defined as $K(l,B)=\frac{1}{d}\sum_{x_0=0}^{d-1}\omega^{-lx_0}B(x_0),$   has the following property: either it is such that:  (A), only for one value of $l$, say $l=s$, one has $K(s,B)\neq 0$, and then $K(s,B)$  is a power of $\omega$, or (B), for every $l$ the value  $K(l,B)$ is some {\em convex} combinations of some subsets of numbers  $\omega^0, \omega, ...,\omega^{d-1}$. 

Proof: The values of $B(x_0)$, are in the form $\omega^{n(x_0)}$, where $n$ is an integer function of $x_0$. Its discrete Fourier transform is 
$K(l,B)=\frac{1}{d}\sum_{x_0=0}^{d-1}\omega^{n(x_0)}\omega^{-lx_0}$, which for every $l$ is exactly such a convex combination. In particular, if the coefficients are not proper convex combinations then only one of them is non-zero. If this is the case for $K(l=s,B)$, then  this is if, and only if,  $B(x_0)=\omega^t\omega^{sx_0}$, where $s,t$ are integers, and $K(l=s,B)=\omega^t$. \qed

\textbf{Lemma 2.} The Fourier expansion coefficients of  powers of  function $B(x_0)$, that is  $B(x_0 )^r$ where $r$ is an integer, have the following form. Assume that for the function   $B$ the Fourier transform 
$K(l=1,B)$  is   in the form of the following convex combination $K(1, B)=\sum_{\nu=0}^{d-1}\lambda_{\nu }\omega^{\nu}$. 
Then the $l=r$ value of the Fourier transform of  $B^r$
 is given by  $K(l=r, B^r)=\sum_{\nu=0}^{d-1}\lambda_{\nu} \omega^{r{\nu}}$. That is, it is a convex combination of $r$-th powers of $\omega^{\nu}$, with the same coefficients $\lambda_{\nu}$, as $K(1,B).$  

Proof:
 The convex combination  coefficients of $K(1, B)$, that is $\lambda_{\nu}$'s, 
in fact,  are equal to $\frac{k_{\nu}}{d}$, where each $k_{\nu}=0,1,..., d$ tells us how many times  the number  $\omega^{\nu}$ 
{appears in $K(1,B)=\frac{1}{d}\sum_{x_0=0}^{d-1}\omega^{n(x_0)}\omega^{-x_0}$}.
As $\omega^{ln(x_0)}\omega^{-lx_0}=(\omega^{n(x_0)}\omega^{-x_0})^l$ and $K(r,B^r)=\frac{1}{d}\sum_{x_0=0}^{d-1}\omega^{rn(x_0)}\omega^{-rx_0}, $ we see that if $\omega^{\nu}$ appears 
$k_{\nu}$ times in $K(1,B)$, so does $\omega^{r\nu}$ in $K(r, B^r).$ \qed

Thus one can replace in  Eq.\eqref{CCP-PAYOFF} the expression $\sum_{x_0=0}^{d-1}\omega^{-lx_0}\tilde{M}_{x}^*(x_0)^l  $  by 
 $K(l, M_x^l)=\sum_{\nu=0}^{d-1}\lambda_{\nu}(x) \omega^{l{\nu}}$. Thus any strategy which is different from the linear one is effectively in terms of payoffs equivalent to a probabilistic strategy in which Alice with probabilities  $\lambda_{\nu}(x)$ chooses the value of the message to be sent to Bob. Such probabilistic strategies are never better than the optimal deterministic one. In the case of a linear strategy  we have situation (A) of Lemma 1,
and thus it is deterministic. Obviously the bound for such a strategy is given by $B$. 
\end{proof}

Let us now move to the quantum strategies which use classical communication, and correlations due to entanglement as a source for information processing, which supplies the partners with partially correlated random noise. The following theorem holds.

\begin{theorem}
The optimal quantum strategy based on classical communication assisted by entanglement for $\mathbb{G}_{\mathscr{B}}$, employs  the linear strategy of messaging,  and achieves its best performance identical to the Tsirelson bound for the associated Bell inequality $\mathscr{B}$.
\end{theorem}

\begin{proof} Using a linear strategy, Bob effectively outputs $G=x_0+a+b\mod{d}$. In order to compute $f_{i,k}$, note that $G=f_{i,k} \Leftrightarrow a+b=F_{xy}^i(k)$. In particular, this strategy eliminates the dependence in Eq.\eqref{perf} on $x_0$. Therefore, the average payoff becomes,
\begin{equation}\label{metric1}
	I_{\mathbb{G}_{\mathscr{B}}}^{cc}=\sum_{x,y}p(x,y)\sum_{k=0}^{K}\sum_{i=1}^{N}c^{i,k}_{xy}P_{xyk}(a+b=F_{xy}^i(k)).
\end{equation}
However, this is precisely the same as the left-hand-side of Eq.\eqref{Bell3}. Since theorem 1 asserts that linear strategies are optimal for implementing  $\mathbb{G}_{\mathscr{B}}$, it follows from Eq.\eqref{Bell3} that
\begin{equation}\label{ccp}
I_{\mathbb{G}_{\mathscr{B}}}^{cc}\leq B,
\end{equation} 
and that the performance in $\mathbb{G}_{\mathscr{B}}$ with a strategy based on classical communication assisted by entanglement can achieve the Tsirelson bound of $\mathscr{B}$.

But is the linear classical messaging strategy also optimal in the entanglement assisted protocol? Notice that in the case of a quantum protocol, we have to introduce an ``answer" observable $\hat{Q}_{}(x_0,x,y)$ of eigenvalues which are powers of $\omega$. This is because Alice, if she follows a deterministic messaging strategy based of her measurement results, the setting of which are determined by her local data $x$, as well as directly on her data, must act as follows. She measures an observable  $\hat{A}(x)$, and if her $i$-th detector fires, she gets an eigenvalue $\xi_i({x})$, whatever it is. Therefore her message will be a function of $\xi_i$ and $x_0,x$, in the form of $m(x_0,x, \xi_i(x))$. But this can be treated a direct measurement of an observable $\hat{m}=m(x_0,x, \hat{A}(x) ) $, which as we know always commutes with $\hat{A}(x)$. Non-degenerate commuting observables differ only by their eigenvalues, but share projectors onto eigenstates. Any degenerate observable can always be put in a form which also has the above features.

The performance of the entanglement assisted protocol is therefore measured by 
\begin{multline}\label{CCP-PAYOFF-Q}
I^{cc}_{\mathscr{B}}=\frac{1}{d}\sum_{x,y}p(x,y)\\
\sum_{x_0=0}^{d-1}\sum_{l=0}^{d-1}\sum_{i=1}^{N}\sum_{k=0}^{K}c^{i,k}_{xy}\omega^{-lF_{xy}^i(k)}\omega^{-lx_0}Tr[\hat{Q}(x_0,x,y)^l\rho^{AB}],
\end{multline}
where $\rho^{AB}$ is the state. With the  assumption that the guess of Bob has the structure $m+b(y)$, the structure of $\hat{Q}_{xy}(x_0)$ must be as follows:
\begin{equation}
\hat{Q}_{xy}(x_0)= \omega^{\hat{m}(x_0,x)}\otimes \omega^{\hat{b}(y)}. 
\end{equation}
The hats denote here local observables of integer eigenvalues. Just as in the classical case, the crucial point is the analysis of the operators given by the expression 
 $\hat{M}(l,x)=\frac{1}{d}
\sum_{x_0=0}^{d-1}\omega^{-lx_0}\omega^{\hat{m}(x_0,x)l}.$
The message observable  $ \omega^{\hat{m}(x_0,x)}$ can be split into sum of projectors $\Pi^i(x)$ multiplied by the associated eigenvalues $\omega^{\eta_i(x_0,x_1,\xi_i)}$. Each $\eta_i$ can be a different function of integer values. This represents the possible strategies of Alice, of how to form the message, once  the result of her measurement of $\hat{m}$ is a collapse of the state given by the  projector $\Pi^i(x)$. This reflects  all possible value assignments to the obtained measurement results, represented by the projectors. Of course  $\sum_{i=1}^{d}\Pi^i(x)=\hat{I}$, where $\hat{I}$ is the local identity operator. With all that, one has  
\begin{equation}\label{AL}
\hat{M}(l,x)=\frac{1}{d}
\sum_{x_0=0}^{d-1}\omega^{-lx_0}\sum_{i=1}^{d}\omega^{l\eta_i(x_0,x,\xi_i)}\Pi^i(x). \end{equation}
Therefore our analysis now moves to the properties of the `effective eigenvalue'  $\frac{1}{d}\sum_{x_0=0}^{d-1}\omega^{-lx_0}\omega^{l\eta_i(x_0,x,\xi_i)}$.
The messaging protocol strategies  are defined the by the structure of the functions $\eta_i$. If one has $\omega^{\eta_i(x_0,x,\xi_i)}\neq \omega^{x_0}\omega^{\xi_i(x)} ,$ then just as in the classical case the effective eigenvalue which survives the summation  over $x_0$ is a convex combination
$\sum_{\nu=0}^{d-1}\lambda_\nu(x,i)\omega^{l\nu}$, where as before $\lambda_\nu =k_\nu/d$ , and  
 $k_{\nu}$ tell us how many times in the sums $\omega^\nu$ is repeated in the sum
$\sum_{x_0=0}^{d-1}\omega^{-lx_0}\omega^{l\eta_i(x_0,x, \xi_i)}$.
The convex combination can be interpreted as a probabilistic mixture of eigenvalues which are powers of $\omega$. Thus, it represents a probabilistic mixture of eigenvalue strategies. However, a mixture of strategies is never better than some deterministic strategy, which thus can give the upper bound of Eq.(\ref{CCP-PAYOFF-Q}). Thus, the eigenvalues should read $\omega^{x_0}\omega^{\xi_i(x)}$. In such a case, our message observable  $ \omega^{\hat{m}(x_0,x)}$ factorizes to  $\omega^{x_0}\sum_{i=1}^{d}\omega^{\xi_i(x)}\Pi^i(x).$ 
We  get a linear strategy 
 and the message, if detector $i$ fires, is $x_0+\xi_i(x)$.  \end{proof}

Essentially, the linear strategy allows us to interpret $x_0$ as a scrambler that Alice uses to randomize her message, as  Bob has no information whatsoever on $a(x)$ for the classical case  or  $\xi_i(x)$ for the  quantum one.  It is never unscrambled, however the linear strategy allows  Bob to guess effectively the functions $f_{i,k}$. This places the original Bell inequality $\mathscr{B}$ and the performance of the linear strategy in $\mathbb{G}_{\mathscr{B}}$ on equal footing: whenever quantum correlations can be used to achieve some value of $I^{bell}$, they can be used to assist classical communication in $\mathbb{G}_{\mathscr{B}}$ such that $I_{\mathbb{G}_{\mathscr{B}}}^{cc}=I^{bell}$, and vice versa.

\textit{Implementing communication complexity reduction protocols  with quantum prepare-transmit-measure strategy.---} We now turn our attention to quantum implementations of $\mathbb{G}_{\mathscr{B}}$ with prepare-transmit-measure protocols. In such a scenario, Alice uses her input data $(x_0,x)$ to prepare a physical state of information content at most $\log d$ bits, i.e., a density matrix $\rho_{x_0x}$ of a $d$-dimensional system. She sends the system to Bob who performs a measurement on it using an observable, the choice of which is dictated by $y$, and obtains an outcome $b_y$. We can easily transform the performance metric of $ \mathbb{G}_{\mathscr{B}}$ in Eq.\eqref{perf} to this alternative implementation in a prepare-transmit-measure scenario. Whenever the output of Bob satisfies $b_{y}=f_{i,k}(x_0,x,y)$  the partnership earns a payoff $c^{i,k}_{xy}$. The average earned payoff is
\begin{eqnarray}\label{ccp2}
&I_{\mathbb{G}_{\mathscr{B}}}^{qc}=\frac{1}{d}\sum_{x_0=0}^{d-1}\sum_{x,y}p(x,y)\nonumber \\&\times
\sum_{i=1}^{N}\sum_{k=0}^{K}c^{i,k}_{xy}P(b_y=f_{i,k}|\rho_{x_0x}, y).&
\end{eqnarray}

Thus, since  $\mathbb{G}_{\mathscr{B}}$ implemented with entanglement and classical communication always can be implemented also with  a prepare-transmit-measure quantum scenario, we can use the considered CCPs as a game-theoretic framework in which we can speak about the two types of quantum protocols  on equal footing.

\textit{Entanglement vs transmission of a quantum system.---}

We make two limiting assumptions: \textbf{AI}, for any Bell inequality $ \mathscr{B}$ we consider situations in which it is violated by quantum predictions, which are achievable with some sets of $d$-outcome measurements of Alice and Bob on entangled systems in a  state $\rho^{AB}\in\mathbb{C}^D\otimes\mathbb{C}^d$ for some integer $D\geq d$, and \textbf{AII}, we consider only such measurements and states used to achieve the maximal violation of the classical bound of the inequality $\mathscr{B}$ for which the following holds:    local measurements of whichever observable in Alice's set gives  uniformly random local results. 

\begin{theorem}
	Assume \textbf{AI} and \textbf{AII}. With any given correlations due to entanglement  we can associate prepare-transmit-measure protocol which  achieves $I_{\mathbb{G}_{\mathscr{B}}}^{cc}=I_{\mathbb{G}_{\mathscr{B}}}^{qc}$.
\end{theorem}

\begin{proof}
We already shown that in quantum theory, for given sets of measurements and a given state,  the maximal value of $I^{bell}$ is the same as that of $I_{\mathbb{G}_{\mathscr{B}}}^{cc}$. Thus, let us  study  quantum violations of the Bell inequality $\mathscr{B}$. Consider the state, $\rho^{AB}$, and the measurements used to achieve a violation of $\mathscr{B}$. The projector  $\Pi^{i}(x)$  of Alice is  associated with  her measurement setting $x$ and her outcome $i$. Now, in the prepare-transmit-measure protocol of $\mathbb{G}_{\mathscr{B}}$, we define the preparations of Alice as the local states of Bob in the Bell scenario after Alice's local measurement, i.e.,  
\begin{equation}\label{rho}
\varrho_{i,x}=d\Tr_A\left(\Pi^{i}(x)\otimes \mathbf{1} \rho^{AB}\right).
\end{equation}
Note that because of assumption \textbf{AI}, communication of the states in Eq.\eqref{rho} is always allowed. Because of assumption \textbf{AII}, we have $p(i|x)=1/d$. Remember that $p(x_0|x)=1/d$ was a premise  when we defined $\mathbb{G}_{\mathscr{B}}$. Therefore in the prepare-transmit-measure scenario we define the set  of Alice's states as $\rho_{x_0x}=\varrho_{i=x_0,x}$. If Bob performs the same measurements as those used to achieve the violation of  $\mathscr{B}$ it follows by construction that there is an analogous violation of Eq.\eqref{ccp2} yielding $I_{\mathbb{G}_{\mathscr{B}}}^{cc}=I_{\mathbb{G}_{\mathscr{B}}}^{qc}$. 


\end{proof}

However, the opposite of Theorem 3 need not be true. Simply by giving some suitable alterations  to some particular  states in the set of preparations  $\{\rho_{x_0x}\}$, we would not be able to reproduce the communicated states by local measurements on an entangled state. This leads to a qualitative relation between the two types of quantum resources: 
\begin{equation}\label{comp}
\text{\textbf{AI} and \textbf{AII}} \Rightarrow I_{\mathbb{G}_{\mathscr{B}}}^{qc}\geq I_{\mathbb{G}_{\mathscr{B}}}^{cc}.
\end{equation}
For any Bell inequality satisfying the given assumptions,  prepare-transmit-measure methods are at least as powerful as  correlations due to entanglement. Of course, from our discussion so far, it is not necessarily the case that a strict inequality can be observed. However, case studies \cite{magic7, TMP16} based on two different families of Bell inequalities satisfying assumptions   \textbf{AI} and \textbf{AII} have revealed multiple such examples. 

However, if we are given a Bell inequalty that does not fulfill both \textbf{AI} and \textbf{AII}, we may find that quantum correlations due to entangled states are  more powerful than prepare-transmit-measure protocols. In fact, for any Bell inequality $\mathscr{B}$ with binary outcomes that achieves its Tsirelson bound by measurements on an entangled state of two $D$-level quantum systems with $D>2$, entanglement is a strictly stronger resource than preparation-transmission-measurement method with  a qubit in $\mathbb{G}_{\mathscr{B}}$. To show this, note that it was shown in Ref.\cite{Hyperbit} that for any CCP with binary answers, entanglement is as least as good a resource as transmission of a single qubit. Note also that $\mathbb{G}_{\mathscr{B}}$ are such CCPs when $\mathscr{B}$ has binary outcomes, i.e., when $d=2$. When the Tsirelson bound of $\mathscr{B}$ is obtained from an entangled state with $D>2$, a strict inequality follows immediately from the fact that the state of Bob after Alice's measurement cannot be reproduced by sending a qubit. Explicit examples of such Bell inequalities in which prepare-transmit-measure protocols are weaker than their entanglement-assisted counterparts have been given in Refs.\cite{PZ10, W10}.

\textit{Discussion.---} We have introduced a game-theoretic framework for studying the ability of quantum correlations obtained from entangled states to assist information processing tasks, as compared to that of prepare-transmit-measure protocols involving only a single quantum system. Importantly, concerning the former resource, we showed that the performance in our CCPs is analogous to the ability of quantum theory to violate a Bell inequality. This opens the door for systematic studies of the comparative nonclassical abilities of the two quantum resources. In particular, we show that for CCPs corresponding to a large class of Bell inequalities, the degree of achievable nonclassicality using a prepare-transmit-measure protocol is as least as much as an entanglement-assisted strategy. Previous case studies \cite{PZ10, magic7, TMP16} further support the potential richness of the relation between the two types of quantum protocols.  Furthermore, the part of our work concerning correlations due to entanglement  can be understood as a generalization of the results of Ref.\cite{BZ} from two-outcome to many-outcome Bell inequalities. Additionally, we presented a proof of the optimality of linear messaging strategies, which was missing in Ref.\cite{BZ}.

From a point of view of possible applications, we note that using our mapping between Bell inequalities and CCPs one can systematically transform many certificates of genuine nonclassical behavior in device independent entanglement assisted protocols to analogous semi-device independent prepare-transmit-measure protocols. Typically, such semi device-independent protocols are somewhat less secure but more efficient than their device independent counterparts. However, due to our relation in Eq.\eqref{comp}, one may obtain further advantages in the efficiency of semi device-independent information processing tasks from the fact that CCPs in a prepare-transmit-measure scheme can to a further extent outperform the classical bound as compared to Bell inequality violations. 

Our work leaves multiple open questions of which we mention some of the more challenging ones: 1) Further qualitative and quantitative characterization of the relation between correlations due to entanglement and protocols based on preparations and measurements of single quantum systems is a key open problem for understanding the extent of nonclassicality enabled by quantum theory, 2) We have only considered bipartite Bell inequalities. Can the mapping between Bell inequalities and CCPs be extended to multipartite scenarios? How will prepare-transmit-measure protocols behave in such scenarios when intermediate partners appear in the chain of communication?, 3) In recent years, much effort has been directed at characterizing Bell-type quantum correlations from information-theoretic principles. Our results suggest that similar attempts to understand the correlations due to single quantum systems may be of interests.

{\em Acknowledgments---} AT acknowledges financial support from the Swiss National Science Foundation (Starting grant DIAQ).
MZ is supported by EU advanced grant QOLAPS, and COPERNICUS grant-award of DFG/FNP.

\end{document}